\documentclass[conference]{IEEEtran}
\usepackage{cite}
\usepackage{amsmath,amssymb,amsfonts,amsthm}
\usepackage{graphicx}
\usepackage{textcomp}
\usepackage{booktabs}
\usepackage{array}
\usepackage{url}
\usepackage{eso-pic}
\AddToShipoutPictureFG*{\AtPageLowerLeft{\put(\LenToUnit{1.6cm},\LenToUnit{0.9cm}){\parbox[b]{\dimexpr\paperwidth-3.2cm\relax}{\centering\footnotesize\itshape This work has been submitted to the IEEE for possible publication. Copyright may be transferred without notice, after which this version may no longer be accessible.}}}}

\renewcommand{\figurename}{Figure}
\renewcommand{\tablename}{Table}
\renewcommand{\thetable}{\arabic{table}}

\renewcommand{\IEEEkeywordsname}{Keywords}
\newcommand{\herm}{^{\mathsf{H}}}
\newcommand{\inv}{^{-1}}
\newcommand{\vect}[1]{\mathbf{#1}}
\newcommand{\SINR}{\mathrm{SINR}}
\newcommand{\opt}{\mathrm{opt}}

\theoremstyle{plain}
\newtheorem{proposition}{Proposition}

\begin{document}

\title{Closed-Form Subspace Loss and Snapshot-Limited Dimension Selection for Beamspace MVDR Radar}

\author{
\IEEEauthorblockN{Moris S. Mukwayi\IEEEauthorrefmark{2}*, Michael A. van Wyk\IEEEauthorrefmark{2}, Andre M. McDonald\IEEEauthorrefmark{3}}
\IEEEauthorblockA{\IEEEauthorrefmark{2}School of Electrical and Information Engineering, University of the Witwatersrand, Johannesburg 2000, South Africa}
\IEEEauthorblockA{\IEEEauthorrefmark{3}Defence and Security Cluster, Council for Scientific and Industrial Research, Pretoria 0184, South Africa}
\IEEEauthorblockA{*Corresponding author: moris.qed@gmail.com}
}

\maketitle

\begin{abstract}
A beamspace minimum variance distortionless response (MVDR) processor adapts a window of
$B$ beams instead of all $N$ elements of a large array, trading a subspace loss that falls
with $B$ against a covariance estimation loss that rises with $B$ for a fixed number $K$ of
training snapshots. We give the subspace loss of a DFT beam window in point jamming in
closed form: exactly for any jammer power, and in the high interference-to-noise ratio
(INR) limit through the Dirichlet kernel of the array. Combined with the known scaled-Beta
law of sample matrix inversion (SMI) with a fixed transform, it gives a window size
$B^\star(K)$ for unloaded SMI. On a simulated $32\times32$ array with carriers from 4 to
8~GHz, noise jammers and Billingsley land clutter, the exact form is within 0.04~dB of the
clairvoyant subspace loss for 25- and 49-beam windows at jammer levels from 10 to 50~dB.
For a 9-beam window whose three columns are taken by three jammers near the target's
elevation, the clutter that the form leaves out costs up to 5~dB. With the nominal $K$,
$B^\star(K)$ picks the lowest-loss window in all 30 cases with more than one runnable
window. The rule is an oracle: it needs the jammer parameters or the true covariance, and
errors of 0.25 bin rms in the jammer directions add a mean loss of 0.04~dB. With diagonal
loading the best window was larger still wherever the rule chose more than the smallest
window.
\end{abstract}


\section{Introduction}

Adaptive beamforming on a large digital array needs an estimate of the
interference-plus-noise covariance. With $N$ elements, sample matrix inversion (SMI) needs
about $2N$ independent training snapshots to come within 3~dB of the optimum \cite{rmb1974},
and the inverse costs of order $N^3$ operations. A ground radar with a $32\times32$ array
would need about 2000 homogeneous range cells in one dwell, and land clutter is seldom
homogeneous over that many cells. The same limit shapes recent work: random-matrix estimators of the inverse
reduced-dimension covariance \cite{song2022} and sparse-Bayesian reduced-dimension space-time
adaptive processing (STAP) from a single training sample \cite{cao2024}. Diagonal loading
\cite{carlson1988,muhammad2022} and reduced-rank estimators \cite{somasundaram2015} lower the
support that the estimate needs. Reduced-dimension processing lowers the size of the problem
instead: a fixed transform maps the array onto a few channels, and the processor adapts there
\cite{vanveen1988,noroozi2025,song2022}. The SINR loss of SMI behind a fixed transform is known to be a
scaled Beta variable \cite{besson2021}. For massive digital arrays, windowed beamspace MVDR on a
spatial FFT comes close to full-dimensional MVDR with much lower complexity \cite{noroozi2025};
its authors studied the trade between window size and FFT resolution in simulation
\cite{noroozi2025} and the window size in tiled beamspace processing \cite{noroozi2025b}. Windowed and tiled
beamspace processing has also been used to cut the cost of massive multiuser MIMO reception
\cite{han2024,cebeci2024}.

How large the window should be for a given number of training snapshots $K$ is usually
settled in simulation \cite{noroozi2025,noroozi2025b}, by the energy that a window captures
from a single source \cite{cebeci2025}, or by dimension-selection rules computed online from
the data \cite{somasundaram2015}. To our knowledge, none gives the subspace loss of a window from the array
and jammer geometry before data are taken, although that loss depends on the full covariance.
Three ingredients here are standard: the projection form of the subspace loss
(Proposition~1) \cite{vanveen1988}, the Woodbury identity that proves Proposition~2, and the
scaled-Beta law of SMI behind a fixed transform (Proposition~3) \cite{rmb1974,besson2021}.
What is new is (i) the closed form \eqref{eq:exact} for the subspace loss of a DFT window in
point jamming, exact at any INR, with its high-INR limit \eqref{eq:limit}, the one-jammer
form \eqref{eq:one} through the Dirichlet kernel of the array, and a bound on the error of
the limit; (ii) the window size $B^\star(K)$ of \eqref{eq:bstar}, which follows from
\eqref{eq:exact} and Proposition~3; and (iii) tests of both at finite INR, in land clutter,
with loading and under array errors, including the equivalent Gaussian support of clutter
training data. $B^\star(K)$ is an oracle design rule, because it needs the jammer parameters
or the true covariance; Section~\ref{sec:sens} tests its sensitivity to errors in them.
The array
changes its carrier from dwell to dwell; frequency agility is used to resist jamming
\cite{zhu2024}, and adaptive beamforming continues to be developed for jammer suppression on
wideband, distributed and frequency-diverse arrays \cite{ma2024,han2023,nysaeter2022}. How the
window follows the carrier is treated in a separate paper by the authors, and here the window is always centred on the exact fractional
bin of the look direction.

\section{Signal Model}

\subsection{Array and interference}
The array is a uniform rectangular array of $N_x\times N_y=32\times32$ elements with pitch
$d=18.75$~mm, half a wavelength at 8~GHz. For direction cosines $(u,v)$ the response of
element $(m,n)$ at carrier $f_h$ is
\begin{equation}
a_{mn}(u,v;f_h)=\exp\!\left(j\,\frac{2\pi f_h d}{c}\,(mu+nv)\right)
\label{eq:steer}
\end{equation}
for $m,n=0,\dots,31$, and $\vect{a}\in\mathbb{C}^{N}$ stacks the $N=1024$ responses of the
look direction. The interference-plus-noise covariance of a range cell is
\begin{equation}
\vect{R}=\sigma_n^2\vect{I}+\vect{A}\vect{D}\vect{A}\herm+\vect{R}_c,
\label{eq:cov}
\end{equation}
where $\vect{A}=[\vect{a}_1,\dots,\vect{a}_J]$ holds the steering vectors of $J$ independent
noise jammers with powers $\vect{D}=\mathrm{diag}(\sigma_1^2,\dots,\sigma_J^2)$, and
$\vect{R}_c$ is the clutter covariance, a sum over patches of the ground ring of the cell.
The training snapshots are drawn from cells around the cell under test and hold no target.

\subsection{MVDR and sample matrix inversion}
The MVDR weight and its output SINR are
\begin{equation}
\vect{w}=\frac{\vect{R}\inv\vect{a}}{\vect{a}\herm\vect{R}\inv\vect{a}},\qquad
\SINR_{\opt}=\sigma_s^2\,\vect{a}\herm\vect{R}\inv\vect{a}.
\label{eq:mvdr}
\end{equation}
SMI replaces $\vect{R}$ by $\hat{\vect{R}}=K^{-1}\sum_r\vect{x}_r\vect{x}_r\herm+\sigma_{\mathrm{DL}}^2\vect{I}$.
Without loading, and for $K\ge N$ independent, zero-mean complex Gaussian training
snapshots, the SINR of the estimated weight divided by the optimum has mean \cite{rmb1974}
\begin{equation}
\mathbb{E}\!\left[\frac{\SINR(\hat{\vect{R}})}{\SINR_{\opt}}\right]=\frac{K+2-N}{K+1}.
\label{eq:rmb}
\end{equation}

\subsection{Beamspace processing}
A beamspace transform $\vect{T}\in\mathbb{C}^{N\times B}$, $B=W^2$, keeps the $W\times W$ DFT
beams around the look direction. Its columns are
\begin{equation}
[\vect{T}]_{(m,n),(k,l)}=\frac{1}{\sqrt{N}}\,e^{\,j2\pi\left[m(p^\star+k)/N_x+n(q^\star+l)/N_y\right]}
\label{eq:T}
\end{equation}
for $k,l=-(W-1)/2,\dots,(W-1)/2$, where $p^\star=u_sN_xdf_h/c$ and $q^\star=v_sN_ydf_h/c$
are the fractional bin indices of the look direction. The columns are orthonormal, and
substituting $p^\star$ and $q^\star$ in \eqref{eq:steer} shows that $\vect{a}$ is $\sqrt{N}$
times column $(0,0)$, so $\vect{a}_b=\vect{T}\herm\vect{a}=\sqrt{N}\,\vect{e}_0$ at every
carrier. The processor adapts $\vect{w}_b=\vect{R}_b\inv\vect{a}_b/(\vect{a}_b\herm\vect{R}_b\inv\vect{a}_b)$
with $\vect{R}_b=\vect{T}\herm\vect{R}\vect{T}$. The following standard result
\cite{vanveen1988} defines the subspace loss, the transform loss of the
reduced-dimension literature.

\begin{proposition}[Subspace loss]
With the true covariance, the SINR of the beamspace processor divided by the element-space
optimum is
\begin{equation}
\rho(\vect{T})=\frac{\vect{a}\herm\vect{T}(\vect{T}\herm\vect{R}\vect{T})\inv\vect{T}\herm\vect{a}}{\vect{a}\herm\vect{R}\inv\vect{a}}
=\frac{\|\vect{P}_{\mathcal S}\vect{R}^{-1/2}\vect{a}\|^2}{\|\vect{R}^{-1/2}\vect{a}\|^2}\le1,
\label{eq:rho}
\end{equation}
where $\vect{P}_{\mathcal S}$ projects onto $\mathcal S=\mathrm{range}(\vect{R}^{1/2}\vect{T})$.
Equality holds if and only if $\vect{R}\inv\vect{a}\in\mathrm{range}(\vect{T})$.
\end{proposition}

\begin{IEEEproof}
Restricting $\vect{w}=\vect{T}\vect{z}$ in the MVDR problem gives the maximum output SINR
$\sigma_s^2\vect{a}_b\herm\vect{R}_b\inv\vect{a}_b$. With $\vect{S}=\vect{R}^{1/2}\vect{T}$ and
$\tilde{\vect{a}}=\vect{R}^{-1/2}\vect{a}$, the numerator of \eqref{eq:rho} is
$\tilde{\vect{a}}\herm\vect{S}(\vect{S}\herm\vect{S})\inv\vect{S}\herm\tilde{\vect{a}}=\|\vect{P}_{\mathcal S}\tilde{\vect{a}}\|^2$.
\end{IEEEproof}

\section{Closed-Form Subspace Loss}

\begin{proposition}[Closed-form subspace loss]
Without clutter, with $\vect{A}_b=\vect{T}\herm\vect{A}$ and $\vect{\Lambda}=\sigma_n^2\vect{D}\inv$,
the subspace loss \eqref{eq:rho} is, for any jammer powers,
\begin{equation}
\rho=\frac{N-\vect{a}_b\herm\vect{A}_b(\vect{A}_b\herm\vect{A}_b+\vect{\Lambda})\inv\vect{A}_b\herm\vect{a}_b}
{N-\vect{a}\herm\vect{A}(\vect{A}\herm\vect{A}+\vect{\Lambda})\inv\vect{A}\herm\vect{a}}.
\label{eq:exact}
\end{equation}
If $c^2<1$, where $c^2=\|\vect{P}_{\mathcal A}\vect{a}\|^2/N$ is the squared cosine between the look direction and
$\mathcal A=\mathrm{range}(\vect{A})$, it tends as the jammer powers grow to the following form, in which the inverses are
pseudo-inverses if $\vect{A}$ or $\vect{A}_b$ lacks full column rank:
\begin{equation}
\rho_\infty=\frac{N-\vect{a}_b\herm\vect{A}_b(\vect{A}_b\herm\vect{A}_b)\inv\vect{A}_b\herm\vect{a}_b}
{N-\vect{a}\herm\vect{A}(\vect{A}\herm\vect{A})\inv\vect{A}\herm\vect{a}}
\label{eq:limit}
\end{equation}
For one jammer it reduces to
\begin{equation}
\rho_\infty=\frac{1-\gamma^2/\eta}{1-\gamma^2},\qquad
\gamma=\frac{|\vect{a}\herm\vect{a}_1|}{N},\qquad
\eta=\frac{\|\vect{T}\herm\vect{a}_1\|^2}{N},
\label{eq:one}
\end{equation}
where $\gamma$ is the quiescent response of the array toward the jammer and $\eta$ the share
of the jammer's energy that the window captures. With the jammer offset $(\Delta p,\Delta q)$
bins from the look direction,
\begin{equation}
\begin{aligned}
\gamma&=|D_{N_x}(\Delta p)|\,|D_{N_y}(\Delta q)|,\\
\eta&=\sum_{k,l}|D_{N_x}(\Delta p-k)|^2\,|D_{N_y}(\Delta q-l)|^2,
\end{aligned}
\label{eq:dirichlet}
\end{equation}
with $D_M(x)=\sin(\pi x)/[M\sin(\pi x/M)]$ and $k,l$ running over the window.
\end{proposition}

\begin{IEEEproof}
By the Woodbury identity, $\vect{a}\herm\vect{R}\inv\vect{a}=\sigma_n^{-2}[N-\vect{a}\herm\vect{A}(\vect{A}\herm\vect{A}+\vect{\Lambda})\inv\vect{A}\herm\vect{a}]$
because $\vect{a}\herm\vect{a}=N$. The transform is orthonormal, so
$\vect{R}_b=\sigma_n^2\vect{I}+\vect{A}_b\vect{D}\vect{A}_b\herm$ and the same identity with
$\|\vect{a}_b\|^2=N$ gives the numerator of \eqref{eq:exact}. Letting all powers grow, the nonzero singular values of $\vect{A}\vect{\Lambda}^{-1/2}$ grow without bound, so the
quadratic forms tend to the squared norms of the projections of $\vect{a}$ and $\vect{a}_b$ onto
$\mathrm{range}(\vect{A})$ and $\mathrm{range}(\vect{A}_b)$, which gives \eqref{eq:limit}. For $J=1$, $|\vect{a}_b\herm\vect{T}\herm\vect{a}_1|^2=N^2\gamma^2$ and
$\|\vect{T}\herm\vect{a}_1\|^2=N\eta$, which gives \eqref{eq:one}; the separable array
factors each inner product into Dirichlet kernels, which gives \eqref{eq:dirichlet}.
\end{IEEEproof}

Form \eqref{eq:exact} costs the same $J\times J$ solves as \eqref{eq:limit} and needs no rank
condition; for one jammer it is $\rho=[1-\gamma^2/(\eta+x)]/[1-\gamma^2/(1+x)]$ with
$x=1/(N\,\mathrm{INR}_1)$. The limit is not uniform: a jammer near a Dirichlet null of the
window adds almost nothing to $\vect{R}_b$ yet enters \eqref{eq:limit} at full weight, and a
jammer at a nonzero integer offset outside the window gives a zero column of $\vect{A}_b$. Let
$\kappa=\lambda_{\min}(\vect{D}^{1/2}\vect{A}_b\herm\vect{A}_b\vect{D}^{1/2})/\sigma_n^2$ (or its
element-space counterpart), and let $c_b^2$ be the squared cosine between $\vect{a}_b$ and $\mathrm{range}(\vect{A}_b)$. In each space the bracket of \eqref{eq:exact} exceeds that of
\eqref{eq:limit} by a factor of at most $1+c^2/[(1-c^2)\kappa]$ (with $c_b$ in the window), which bounds the error of the
limit. The limit therefore needs two things: $c^2<1$, and $\kappa$ large against $c^2/(1-c^2)$. As $c^2\to1$ the optimum is nulled with the jammers
and \eqref{eq:limit} becomes $0/0$ (a jammer on the look direction gives $\rho=1$ from \eqref{eq:exact} at any INR). Rank deficiency of $\vect{A}_b$ makes
$\kappa=0$ and voids the bound, although the pseudo-inverse form is still the limit. Numerically, a singular value below $10^{-8}$ of the largest is set to zero in a
pseudo-inverse, and $\vect{A}_b$ is called nearly singular when the smallest singular value of $\vect{A}_b/\sqrt{N}$ is below $3\times10^{-5}$ of the largest (Section~\ref{sec:closed}).

The centre beam alone contributes $\gamma^2$ to $\eta$, and the window cannot capture more than
the whole jammer, so $\gamma^2\le\eta\le1$ and \eqref{eq:one} lies between 0 and 1. Since
$1-\gamma^2=(\eta-\gamma^2)+(1-\eta)$, the loss is small when $1-\eta\ll\eta-\gamma^2$, that
is, when the window holds much more of the jammer's energy outside the centre beam than the
array receives outside the window, as for a jammer near the centre of one of the outer beams.
A small $\gamma$ alone does not make the loss small: far outside the window $\gamma^2$ and
$\eta$ are both small, the jammer's leakage is spread almost evenly over the window beams
along each direction in which it is far away, and $\gamma^2/\eta$ tends to between $1/W^2$ and
$1/W$, the latter for a jammer in the target's row or column. Such a jammer still has a
limiting loss as low as $10\log_{10}(1-1/W)$: $-1.8$, $-1.0$ and $-0.7$~dB for $W=3$, 5 and 7.
The loss is largest for a jammer within a fraction of a bin of the look direction (as low as
$-2.1$, $-1.2$ and $-0.8$~dB). With several jammers the losses do not add, and
\eqref{eq:exact} or \eqref{eq:limit} has to be evaluated jointly.

Besides point sources and a look direction outside the jammer span, the closed forms assume an ideal, separable
manifold, with no coupling or channel error, and one carrier for the whole dwell. Clutter is
spread in angle, so in clutter the subspace loss has to be computed from \eqref{eq:rho};
Section~\ref{sec:closed} measures what leaving it out of \eqref{eq:exact} costs.

\section{Loss Decomposition and Window Size}

\begin{proposition}[Loss decomposition \cite{besson2021}]
Let the $K\ge B$ training snapshots be independent, zero-mean complex Gaussian with covariance
$\vect{R}$. For unloaded SMI in the beamspace, $\SINR_b/\SINR_{\opt}$ is distributed as
$\rho(\vect{T})\,\beta$ with $\beta\sim\mathrm{Beta}(K+2-B,\,B-1)$, so that
\begin{equation}
\mathbb{E}\!\left[\frac{\SINR_b}{\SINR_{\opt}}\right]=\rho(\vect{T})\,\frac{K+2-B}{K+1}.
\label{eq:decomp}
\end{equation}
\end{proposition}

\begin{IEEEproof}
$\SINR_b/\SINR_{\opt}$ is $\rho(\vect{T})$ times the loss of SMI relative to the clairvoyant
beamspace processor. The transform does not depend on the data, so the beamspace snapshots are
independent, zero-mean complex Gaussian with covariance $\vect{R}_b$, and that loss is the
Beta variable of \cite{rmb1974} in $B$ dimensions.
\end{IEEEproof}

In decibels, the loss of the mean in \eqref{eq:decomp} is the sum of the subspace loss and the
estimation loss; the mean of the decibel values is not. Maximising the product over the
candidate windows gives
\begin{equation}
B^\star(K)=\arg\max_B\ \rho(B)\,\frac{K+2-B}{K+1},
\label{eq:bstar}
\end{equation}
where $\rho(B)$ comes from \eqref{eq:exact} in point jamming or from \eqref{eq:rho} otherwise.
$B^\star(K)$ is the optimum for unloaded SMI under the assumptions of Proposition~3, and it is an oracle design
rule: $\rho(B)$ needs the jammer directions and powers, through \eqref{eq:exact}, or the true covariance,
through \eqref{eq:rho}, and Table~\ref{tab:choice} evaluates it with the clairvoyant $\rho(B)$.

Proposition~3 also assumes
that $\vect{T}$ does not depend on the training data. That holds here: the position of the window follows from the look direction and the
carrier, and its size from $K$ and $\rho(B)$, none of which is read from the training snapshots. A selector
that estimated $\rho(B)$ from the snapshots that also estimate $\vect{R}_b$ would make $\vect{T}$
data-dependent, so that \eqref{eq:decomp} would no longer follow, and an in-sample estimate would be expected to favour larger
windows. Splitting the snapshots restores the assumption at the price of a smaller $K$ in \eqref{eq:decomp}. Neither selector
is tested here.

The window also sets the computing load: the beamspace solve costs about $B^3/3+KB^2$ operations per
weight update, $3.6\times10^4$ for $B=25$ at $K=50$, against $4.1\times10^8$ for direct
element-space inversion ($N^3/3+KN^2$ with $N=1024$).

Clutter training data are not independent Gaussian snapshots: the texture changes the power
from cell to cell, and pulses of one cell are correlated. We measure the effect on the sample
covariance with an equivalent Gaussian support $K_{\mathrm{eff}}$, the number of independent
Gaussian snapshots for which \eqref{eq:rmb}, with $N$ replaced by $B$, gives the measured
estimation loss $L$ (the SINR of unloaded SMI divided by that of the clairvoyant beamspace
processor on the same dwell):
\begin{equation}
K_{\mathrm{eff}}=\frac{B-2+L}{1-L}.
\label{eq:keff}
\end{equation}
$K_{\mathrm{eff}}$ is inferred from the observed loss and is not a count of independent
snapshots; Monte Carlo scatter can put it slightly above $K$.

Diagonal loading also departs from Proposition~3. It keeps the noise eigenvectors of the
estimate from distorting the pattern \cite{carlson1988}, and with loading between the noise
and the interference eigenvalues the estimation loss depends on the interference rank in the
window instead of on $B$ \cite{besson2021}. A larger window then costs little once that rank
stops growing, and \eqref{eq:bstar} says nothing exact about the best loaded window.

\section{Simulation Setup}

The results come from a Simulink testbench built from library blocks. Each environment model
was checked against an independent reference: SMI reproduces \eqref{eq:rmb} at $K=2N$ to
0.002~dB, and the Simulink model matches a separate MATLAB reference on the same random draws
to within $1.5\times10^{-8}$ (relative). Table~\ref{tab:scen} lists the two scenarios.

Scenario P is a jamming-only reference, and scenario G a ground radar over low-relief rural
land. Clutter patches $1^\circ$ wide over $\pm90^\circ$ carry Billingsley's reflectivity and
his Rice-distributed intrinsic-motion model \cite{billingsley2002}, and a unit-mean texture
sets the power from cell to cell: constant (Rayleigh clutter), Weibull with $a_w=2$ (texture
variance 5) or Gamma with $\nu=1$, which gives K-distributed clutter
\cite{conte1995}. The texture can be held over blocks of 4 or 16 range cells. The
jammers reach the array along a direct path and its ground image. In G2 the target is
$5^\circ$ above the horizon, the jammers are on, and the clutter ridge of the cell under test
lies inside the window. In G1 the target is $1.32^\circ$ up, the ridge lies in the main beam
about 5.5~dB below the noise per element, and the jammers are switched off so that clutter
alone sets the training statistics. All scenario G results use the snapshots before Doppler
filtering: a training snapshot is the array vector of the first pulse of one range cell, or of
its first $M_p$ pulses when pulses are pooled.

\begin{table}[!t]
\caption{Scenarios. Powers are per element and relative to the receiver noise of one element.}
\label{tab:scen}
\centering
\footnotesize
\setlength{\tabcolsep}{3pt}
\begin{tabular}{@{}p{0.48in}p{1.18in}p{1.62in}@{}}
\toprule
 & Scenario P & Scenario G \\
\midrule
Array & \multicolumn{2}{l}{$32\times32$ URA, $d=18.75$~mm, carrier 4 to 8~GHz} \\
Element & isotropic & $\cos^{1.5}\theta$ (power), none behind the face \\
Radar & not modelled & 15~m mast; $M=16$ pulses, PRF 1.5~kHz, 5~MHz (30~m cells) \\
Target & $(u_s,v_s)=(0.55,0)$; also $u_s=0.20$, 0.85 & 12~km, $u_s=0.55$; elevation $5^\circ$ (G2) or $1.32^\circ$ (G1) \\
Jammers & 3 at 30~dB: $(0.675,0)$, $(0.283,0.05)$, $(0.600,0.35)$ & 3 at 100~km and 8~km altitude, $u=0.675$, 0.283, 0.900; 30~dB (G2), 10 to 50~dB in Table~\ref{tab:limit}; none in G1 \\
Clutter & none & Billingsley low-relief rural, wind 4.5~m/s; Rayleigh, Weibull ($a_w=2$) or K ($\nu=1$) texture \\
Training & $K=18$, 50, 98 & $K=16$ to 128 \\
Windows & \multicolumn{2}{l}{$B=9$, 25, 49 ($W=3$, 5, 7); loading $\sigma_{\mathrm{DL}}^2=1$ when used} \\
\bottomrule
\end{tabular}
\end{table}

Each window runs unloaded SMI (when $K>B$), loaded SMI with $\sigma_{\mathrm{DL}}^2$ equal to the
noise power, and the clairvoyant beamspace processor, whose loss against the element-space
optimum is the clairvoyant subspace loss $\rho(\vect{T})$ of \eqref{eq:rho} with the
simulator's covariance. Element-space SMI with the same loading on the same snapshots is the
reference. Losses are averaged over dwells in linear units and shown in decibels. The scenario
P runs hold the carrier at each of 41 frequencies from 4 to 8~GHz for 500 dwells, and the
closed-form test holds it at each of 21 frequencies for 20 dwells; the scenario G runs of
Figure~\ref{fig:est} and Table~\ref{tab:choice} hop pseudo-randomly over the band, 500 dwells
per run.

\section{Results}

\subsection{Closed-form subspace loss}\label{sec:closed}
Table~\ref{tab:limit} gives the error of the limit \eqref{eq:limit} and of the exact form
\eqref{eq:exact}, both without clutter, against the clairvoyant subspace loss in G2. For $B=25$
and 49 the exact form is within 0.04~dB at every INR and clutter level, and the limit is within
0.05~dB from 20~dB up and 0.26~dB at 10~dB. The band-mean clairvoyant loss hardly moves with
the clutter power: $-0.65$ to $-0.69$~dB for $B=25$ and $-0.27$~dB for $B=49$.

\begin{table}[!t]
\caption{Worst-carrier error (\textnormal{dB}) of the clutter-free limit (8) and exact form (7) against the clairvoyant subspace loss with clutter, G2, 21 carriers from 4 to 8~\textnormal{GHz}. Each entry is the range over five cells under test, 6 to 15~\textnormal{km} (clutter-to-noise ratio $-7$ to $-19$~\textnormal{dB} per element).}
\label{tab:limit}
\centering
\footnotesize
\setlength{\tabcolsep}{3pt}
\begin{tabular}{@{}cl ccccc@{}}
\toprule
 & & \multicolumn{5}{c}{Jammer INR (dB)} \\
\cmidrule(l){3-7}
$B$ & Form & 10 & 20 & 30 & 40 & 50 \\
\midrule
9  & limit & 12.4 & 9.9 & 8.5--8.6 & 7.6 & 4.4--5.0 \\
9  & exact & 0.5--0.6 & 2.7--3.0 & 4.1--4.7 & 4.4--5.0 & 4.4--5.0 \\
25 & limit & 0.24--0.26 & 0.03 & 0.02--0.04 & 0.02--0.04 & 0.02--0.04 \\
25 & exact & 0.02--0.04 & 0.02--0.04 & 0.02--0.04 & 0.02--0.04 & 0.02--0.04 \\
49 & limit & 0.04--0.05 & 0.02 & 0.02 & 0.02 & 0.02 \\
49 & exact & 0.02 & 0.02 & 0.02 & 0.02 & 0.02 \\
\bottomrule
\end{tabular}
\end{table}

In the $3\times3$ window, three jammers at nearly the target's elevation use up the window. All
three lie 0.16 bins below the target in $v$ at 6~GHz, so more than 90\% of the energy they
deliver to the window falls in the target's row, at three different offsets in $u$ (1.5,
$-3.2$ and 4.2 bins at 6~GHz). Nulling them takes the three columns of that row: the limit
gives $-10.1$ to $-16.2$~dB across the band, against $-1.0$ to $-4.3$~dB with any one jammer
removed, and on its own the jammer nearest the window edge has a loss of only $-0.24$~dB at
6~GHz. The two outer jammers leak only weakly into the window ($\eta\approx0.01$ at 6~GHz), so
the limit is reached slowly: the indicator
$\lambda_{\min}(\vect{D}^{1/2}\vect{A}_b\herm\vect{A}_b\vect{D}^{1/2})/\sigma_n^2$ is only 0.8 at
6~GHz and 10~dB INR and 79 at 30~dB. At 10~dB INR the exact form gives a band mean of $-3.6$~dB
against $-3.7$~dB for the clairvoyant loss. At higher INR the clutter ridge in the lower row of
the window, absent from \eqref{eq:exact}, adds loss once the jammers are nulled: at 50~dB the
clairvoyant loss is $-13.6$ to $-14.6$~dB against $-11.9$~dB from the exact form, and deeper
than the clutter-free limit at 20 of the 21 carriers.

A random-laydown screen for an ideal array without clutter tests the limit more widely:
$10^4$ draws of one to five jammers uniform over the visible region, with INRs uniform in
decibels from 10 to 50~dB and carriers uniform from 4 to 8~GHz, target as in G2. The limit is
computed with pseudo-inverses, so it exists wherever $c^2<1$; the largest $c^2$ was 0.984. Draws
in which $\vect{A}_b$ is nearly singular (smallest singular value of $\vect{A}_b/\sqrt{N}$ below
$3\times10^{-5}$ of the largest: 403, 182 and 108 of the $10^4$ for $W=3$, 5 and 7) have
$\kappa<0.03$, so that \eqref{eq:exact} is nowhere near its limit. Against the exact form, the
limit has a median error of 0.17, 0.02 and 0.002~dB and a 95th percentile of 1.79, 0.24 and
0.08~dB on the other draws, and of 1.86, 0.25 and 0.08~dB with them. Wherever $\kappa$
exceeded 100 (13 to 35\% of the draws, by window), the error was at most 0.02~dB; for $W=3$
the median error is 0.44~dB where $\kappa<1$ and 0.07~dB where $1\le\kappa<10$. In the G2
$3\times3$ window, where $c^2$ is 0.95 to 0.98, the same $\kappa=100$ still leaves an error of
0.5 to 1.1~dB.

We also evaluated \eqref{eq:limit} with the ideal manifold on data that pass through a
perturbed one, over 500 hopping dwells in P and in G2. Under mutual coupling between neighbours
within two element spacings ($\epsilon=0.1$), random gain and phase errors of 0.5~dB and
$5^\circ$ rms, and an outer ring of elements detuned by $-1.5$~dB and $+10^\circ$, the median
error over the carriers stayed within 0.14~dB in both scenarios. Spreading each jammer over
five tones 1~MHz apart, each with a fifth of its power, gives the jammers of P a second
eigenvalue 3.5 to 11~dB above the noise across the aperture, and the rank-one model was then
0.15~dB off in P and 1.68~dB off in G2. With the look direction 0.2 beamwidth off the target,
the median error was 0.28~dB in P and 0.46~dB in G2.

\subsection{Loss decomposition}
Scenario P meets the conditions of Proposition~3. Table~\ref{tab:decomp} compares its measured
loss of unloaded SMI at $K=2B$ with the product \eqref{eq:decomp} of the clairvoyant subspace
loss and the estimation factor. The band means agree within 0.02~dB for all three windows and
all three target azimuths. The standard error of a 500-dwell mean of the unloaded SMI loss at
$K=2B$ is 0.037, 0.025 and 0.019~dB for $B=9$, 25 and 49 at one carrier, so the band-mean
differences are within the Monte Carlo noise.

\begin{table}[!t]
\caption{Loss decomposition (11) against simulation: unloaded SMI at $K=2B$, scenario P, band means over 4 to 8~\textnormal{GHz} (\textnormal{dB}). Factor: $(K+2-B)/(K+1)$.}
\label{tab:decomp}
\centering
\footnotesize
\begin{tabular}{@{}cccrrrr@{}}
\toprule
$u_s$ & $B$ & $K$ & Subspace & Factor & Predicted & Measured \\
\midrule
0.55 & 9  & 18 & $-1.34$ & $-2.37$ & $-3.71$ & $-3.70$ \\
0.55 & 25 & 50 & $-0.24$ & $-2.76$ & $-3.00$ & $-3.00$ \\
0.55 & 49 & 98 & $-0.05$ & $-2.88$ & $-2.93$ & $-2.93$ \\
0.20 & 9  & 18 & $-2.13$ & $-2.37$ & $-4.50$ & $-4.52$ \\
0.20 & 25 & 50 & $-0.92$ & $-2.76$ & $-3.68$ & $-3.68$ \\
0.20 & 49 & 98 & $-0.45$ & $-2.88$ & $-3.33$ & $-3.34$ \\
0.85 & 9  & 18 & $-1.08$ & $-2.37$ & $-3.45$ & $-3.45$ \\
0.85 & 25 & 50 & $-0.36$ & $-2.76$ & $-3.13$ & $-3.13$ \\
0.85 & 49 & 98 & $-0.16$ & $-2.88$ & $-3.04$ & $-3.03$ \\
\bottomrule
\end{tabular}
\end{table}

\subsection{Estimation loss and equivalent support in clutter}
Figure~\ref{fig:est} gives the estimation loss against $K$ for the three windows, with the
jammers on (G2) and with clutter alone (G1), from which $K_{\mathrm{eff}}$ follows by
\eqref{eq:keff}. With the jammers on, unloaded SMI follows the factor of \eqref{eq:rmb} for
every texture. For $B=25$ the measured loss is $-5.62$, $-2.94$, $-2.03$, $-1.23$ and
$-0.88$~dB at $K=32$, 48, 64, 96 and 128, against $-5.64$, $-2.92$, $-2.00$, $-1.23$ and
$-0.89$~dB predicted, and Weibull and K-distributed clutter change these by at most 0.05~dB.
The equivalent support of the $5\times5$ window is within 2\% of $K$ for every texture, and it
stays there when the texture is held over blocks of 4 or 16 range cells. The 30~dB jammers are
white from cell to cell and dominate the training data, so the clutter texture has little
effect on the estimate, except in the $3\times3$ window, where the jammers' leakage is weak:
its equivalent support falls to $0.80K$ at $K=48$ and $0.63K$ at $K=128$ in Weibull clutter
and to $0.86K$ at $K=128$ in K-distributed clutter.

\begin{figure}[!t]
\centering
\includegraphics[width=0.84\columnwidth]{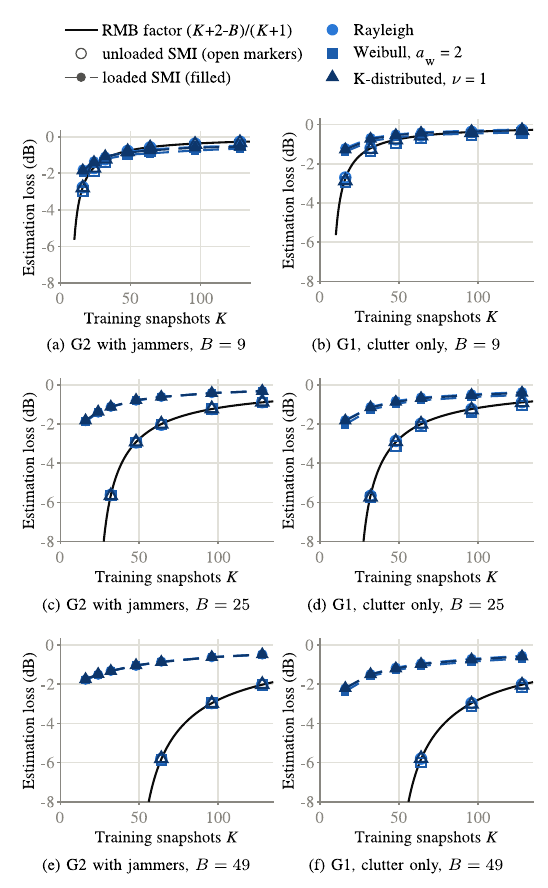}
\caption{Estimation loss (SMI loss divided by the subspace loss of the same dwell) against the number of training snapshots for the three textures. Open markers: unloaded SMI; filled markers and dashed lines: loaded SMI; solid line: the factor $(K+2-B)/(K+1)$.}
\label{fig:est}
\end{figure}

With clutter alone, Rayleigh and K-distributed clutter still follow \eqref{eq:rmb}: for $B=25$
the equivalent support stays within 3\% of $K$ from $K=32$ to 128. Weibull clutter, whose
texture variance is five times that of the K-distributed case, costs support as $K$ grows. For
$B=25$, $K_{\mathrm{eff}}/K$ is 0.99 at $K=32$, 0.96 at 48 and 64, 0.91 at 96 and 0.87 at 128;
for $B=9$ it falls to 0.64 at $K=128$. The cost in decibels is small, because support is lost
only where the estimation loss is already small: at $K=128$ the estimation loss of unloaded SMI
with $B=25$ is $-1.04$~dB in Weibull clutter against $-0.88$~dB in Rayleigh clutter. Holding the
K-distributed texture over 4 or 16 cells gives a support of 47.3 and 47.9 at $K=48$ (48.3 with
independent cells).

A radar short of range cells can pool several pulses of each cell. We estimate the pooled
support from $K=16$ cells and the unloaded estimation loss. With the jammers on, 4 or 16 pooled
pulses give the $5\times5$ window 98 to 99\% of the nominal support $KM_p$ at wind speeds of
1.3, 4.5 and 9.0~m/s, and the $3\times3$ window 74 to 76\% from 16 pulses. With clutter alone
the $5\times5$ window gets 98\% from 4 pulses but only 86 to 90\% from 16, least at the lowest
wind speed, and the $3\times3$ window 70 to 74\% from 16. The pulses of one cell see nearly the
same clutter, whose spectrum at 4.5~m/s wind and 6~GHz is a few hertz wide against a 94~Hz
Doppler bin, so pooling adds independent noise samples faster than independent clutter samples.
Loading at the noise power cuts the estimation loss sharply (Figure~\ref{fig:est}, filled
markers): for $B=25$ in G2 it is $-1.85$~dB at $K=16$, where unloaded SMI cannot run at all,
$-0.80$~dB at $K=48$ and $-0.31$~dB at $K=128$. With clutter alone in G1, the total loss of
loaded SMI with $B=25$ is $-2.52$~dB at $K=16$ and $-1.09$~dB at $K=128$, against $-10.66$ and
$-2.83$~dB for loaded element-space SMI on the same Rayleigh-clutter snapshots.

\subsection{Choosing the window}
Figure~\ref{fig:win} shows the trade between window size and training support in scenario P.
Without loading the estimation factor dominates and the best window shrinks as $K$ falls: at
$u_s=0.55$ and $K=50$ the losses are $-2.08$, $-3.00$ and $-12.38$~dB for $B=9$, 25 and 49,
while at $K=98$, $B=25$ is best ($-1.44$~dB against $-1.70$ and $-2.93$~dB). With loading the
two larger windows always beat $B=9$. At $K=50$ the band means are $-1.71$, $-0.88$ and
$-0.94$~dB for $B=9$, 25 and 49; at this azimuth $B=25$ and 49 stay within 0.16~dB of each
other at every $K$, and $B=9$ trails them by 0.7 to 1~dB. Loaded element-space SMI with all
1024 channels (dotted) has losses of $-0.45$ to $-0.80$~dB.

\begin{figure}[!t]
\centering
\includegraphics[width=0.84\columnwidth]{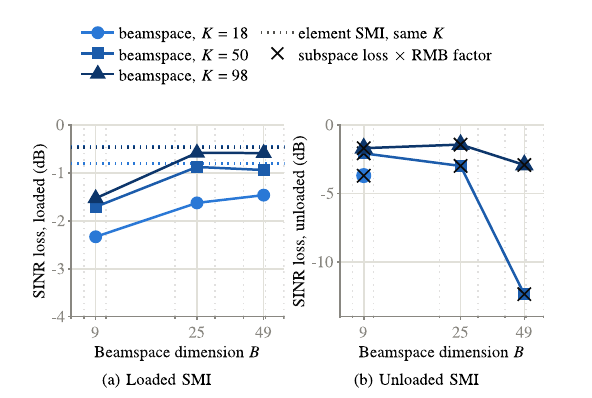}
\caption{Band-mean SINR loss against window size, scenario P, $u_s=0.55$. In (a) the dotted lines give loaded element-space SMI on the same snapshots; in (b) the crosses give the product \eqref{eq:decomp}.}
\label{fig:win}
\end{figure}

Table~\ref{tab:choice} tests \eqref{eq:bstar} as an oracle rule in scenario G, with the nominal $K$ and the
clairvoyant subspace loss. With the nominal $K$, $B^\star(K)$ picks the best window for unloaded
SMI in all 39 cases, and in all 30 cases where more than one window can run; using the
equivalent support of the $5\times5$ window instead gives the same choices. Many of these
decisions are settled in advance. In 9 cases ($K\le24$ with the jammers on, $K=16$ with clutter
alone) only the $3\times3$ window can run, and in the 15 cases with the jammers on and
$K\ge32$ the $3\times3$ window starts about 10~dB down. The informative cases are the 15 with
clutter alone and $K\ge32$: the three windows have subspace losses of $-1.3$, $-0.7$ and
$-0.5$~dB, the smaller estimation loss of the $3\times3$ window keeps it best, and the margin
shrinks from 3.8 to 3.9~dB at $K=32$ to a tie at $K=128$ ($-1.60$~dB). The margins are
differences of 500-dwell means with standard errors of a few hundredths of a decibel, so only
that tie is within the Monte Carlo noise.

\begin{table}[!t]
\caption{Oracle window size (12) with the nominal $K$, the window with the smallest measured loss, and the margin (\textnormal{dB}) between it and the next. $K_{\mathrm{eff}}$: equivalent support of the $5\times5$ window. Where the textures disagree the entry lists Rayleigh / Weibull / K-distributed clutter.}
\label{tab:choice}
\centering
\footnotesize
\setlength{\tabcolsep}{2.5pt}
\begin{tabular}{@{}r l c c c c@{}}
\toprule
$K$ & $K_{\mathrm{eff}}$, $B=25$ & $B^\star$ & Best, unl. & Margin & Best, loaded \\
\midrule
\multicolumn{6}{@{}l}{\textit{G2, jammers at 30 dB}} \\
16  & n/a & 9 & 9 & n/a & 49 \\
24  & n/a & 9 & 9 & n/a & 49 \\
32  & 32.1 / 32.0 / 32.0 & 25 & 25 & 4.3--4.5 & 49 \\
48  & 47.8 / 47.6 / 48.1 & 25 & 25 & 7.2--7.6 & 49 \\
64  & 63.2 / 63.5 / 63.9 & 25 & 25 & 3.4--3.5 & 49 \\
96  & 96.2 / 96.1 / 97.1 & 25 & 25 & 1.3--1.4 & 49 \\
128 & 129.4 / 127.4 / 130.1 & 25 & 25 & 0.7 & 49 \\
\midrule
\multicolumn{6}{@{}l}{\textit{G1, clutter only}} \\
16  & n/a & 9 & 9 & n/a & 9 / 9 / 25 \\
32  & 31.9 / 31.7 / 31.8 & 9 & 9 & 3.8--3.9 & 25 \\
48  & 48.7 / 45.9 / 48.3 & 9 & 9 & 1.5--1.6 & 25 \\
64  & 64.4 / 61.3 / 63.9 & 9 & 9 & 0.8 & 25 \\
96  & 95.6 / 87.3 / 93.4 & 9 & 9 & 0.25--0.26 & 49 / 25 / 49 \\
128 & 129.4 / 111.7 / 126.2 & 9 & 9 & 0.00--0.02 & 49 \\
\bottomrule
\end{tabular}
\end{table}

The loaded processor prefers $B=49$ at every $K$ with the jammers on, and $B=25$ or 49 from
$K=32$ with clutter alone, as the interference-rank argument of Section~IV suggests. In the 15
cases where $B^\star$ is more than the smallest window (all with the jammers on), the best
loaded window is larger still, and in G1 and G2 the 49-beam window is never more than 0.2~dB
worse than the best.

\subsection{Sensitivity of the rule to its inputs}\label{sec:sens}
Rule \eqref{eq:bstar} needs $\rho(B)$, and so the jammer directions, powers and number. In scenario P, where
\eqref{eq:exact} and \eqref{eq:decomp} are exact, the effect of errors in these inputs can be computed
without simulation; the closed forms reproduce the subspace-loss column of Table~\ref{tab:decomp} to within 0.004~dB. For each of the three look directions, 11 carriers and $K=32$, 50, 64, 98 and 128, 200 draws
perturbed the inputs of \eqref{eq:exact} before \eqref{eq:bstar} was evaluated. The regret is the true loss of
the window chosen relative to that of the best window (Table~\ref{tab:sens}). With errors of 0.25 bin rms in each
direction cosine the rule still chooses the best window in 86\% of the cases, the mean regret is 0.04~dB and the largest
is 1.6~dB. Errors of $\pm10$~dB in each INR hardly matter, and missing one of the three jammers gives 83\% and a largest regret of
1.3~dB. The direction of the jammers is the input that matters most. In this scenario the decision is easy, however: a fixed
$3\times3$ window already has a mean regret of 0.14~dB, against 0.99~dB for a fixed $5\times5$ window. This tests the rule against perturbed jammer parameters. The route through the true covariance, used in clutter, is not assessed, and nothing here estimates the inputs from data.

\begin{table}[!t]
\caption{Effect of errors in the inputs of \eqref{eq:bstar}, scenario P, unloaded SMI (33\,000 draws per row unless stated). Best: share of draws in which the rule chooses the best window. Regret: loss of the chosen window relative to the best.}
\label{tab:sens}
\centering
\footnotesize
\setlength{\tabcolsep}{3pt}
\begin{tabular}{@{}lrrrr@{}}
\toprule
 & & \multicolumn{3}{c}{Regret (\textnormal{dB})} \\
\cmidrule(l){3-5}
Perturbation of the inputs & Best (\%) & Mean & 95th & Max \\
\midrule
Directions, 0.25 bin rms & 85.8 & 0.04 & 0.31 & 1.64 \\
Directions, 1 bin rms & 78.7 & 0.07 & 0.46 & 3.40 \\
INRs, $\pm10$~\textnormal{dB} & 99.5 & 0.00 & 0.00 & 0.33 \\
One jammer missed & 83.0 & 0.05 & 0.42 & 1.27 \\
One extra jammer & 95.2 & 0.01 & 0.00 & 3.40 \\
Limit \eqref{eq:limit} used (165 cases) & 98.8 & 0.00 & 0.00 & 0.43 \\
\midrule
Fixed $3\times3$ window (165 cases) & n/a & 0.14 & 0.68 & 1.27 \\
Fixed $5\times5$ window (165 cases) & n/a & 0.99 & 3.56 & 4.06 \\
\bottomrule
\end{tabular}
\end{table}

\section{Discussion}

To choose a window, evaluate \eqref{eq:exact} with all the jammers together for each candidate.
The single-jammer form \eqref{eq:one} misses windows whose columns or rows are used up by
several jammers near the target's elevation or azimuth, and the limit \eqref{eq:limit} is safe
only when $\kappa$ is large compared with $c^2/(1-c^2)$. For unloaded SMI, take $B^\star(K)$
from \eqref{eq:bstar}, an oracle rule whose choice changed little under errors in its inputs (Table~\ref{tab:sens}). For SMI loaded at the noise power, take at least $B^\star(K)$; the
49-beam window cost at most 0.2~dB more than the best, but loaded element-space SMI lost up to
8.1~dB more than $B=25$ with clutter alone, so the computing budget alone should not set the
upper end.

Propositions~2 and 3 are exact under their assumptions: \eqref{eq:exact} for point jamming with
an ideal manifold, and \eqref{eq:decomp} for independent Gaussian training. Everything else is
a numerical observation from one laydown per scenario, one array size and one loading level.
The clutter levels are modest (5.5~dB below the noise per element in G1), and stronger clutter
would make the clutter term that \eqref{eq:exact} leaves out matter more. The rank-one form
should not be used for jammers spread over the receive band, whose second eigenvalue it
ignores.

\section{Conclusion}

For unloaded SMI in the scenarios tested, the window size followed from the subspace loss and
the training support: the closed form \eqref{eq:exact} gave the subspace loss of the 25- and
49-beam windows within 0.04~dB, and $B^\star(K)$ chose the best window in all 30 cases with more
than one runnable window, although many of these were settled in advance. For SMI loaded at the
noise power the best window was never smaller than $B^\star(K)$. The rule is an oracle, and it was robust to errors in its inputs, least to those in the jammer directions. An estimator of its
inputs from the training snapshots, with the data dependence that this brings, remains to be built, and the interference-rank form of
the loaded SMI loss \cite{besson2021} is the natural basis for a window rule for loaded
estimators.

\section*{Acknowledgment}
The authors used Claude (Anthropic), a large language model, to help prepare this version of the paper from their earlier draft. The authors provided the results and the derivations; Claude helped to format the paper for the conference template and to proof-read it. The authors reviewed and accepted all changes, approved the final text, and take full responsibility for the content.

\end{document}